\pdfoutput=1

\documentclass{article}
\usepackage{amssymb,amsmath,amsthm}
\usepackage{enumerate}
\usepackage{epsfig}
\usepackage{xspace}

\newtheorem{theorem}{Theorem}[section]

\newtheorem{definition}[theorem]{Definition}

\newtheorem{lemma}[theorem]{Lemma}
\newtheorem{polyrule}{Rule}[section]

\DeclareMathOperator{\operatorClassNP}{NP}
\newcommand{\classNP}{\ensuremath{\operatorClassNP}}

\DeclareMathOperator{\operatorClassFPT}{FPT}
\newcommand{\classFPT}{\ensuremath{\operatorClassFPT}}

\def\etal{{\em et al.}}

\usepackage{mathrsfs}

\usepackage{vmargin}
\setmarginsrb{1.5in}{1.5in}{1.5in}{1.5in}{0cm}{0.0cm}{0cm}{1cm}

\newcommand{\fas}{feedback arc set\xspace}
\newcommand{\FAST}{\textsc{$k$-FAST}\xspace}
\newcommand{\wFAST}{$k$-\textsc{WFAST}\xspace}

\title{Kernels for Feedback Arc Set In Tournaments}
\author{
 St{\'e}phane Bessy\thanks{%
  LIRMM -- Universit{\'e} de Montpellier 2, CNRS,
  161 rue Ada,
  34392 Montpellier,
  France.
  \texttt{\{bessy|paul|perez|thomasse\}@lirmm.fr}}
 \and Fedor V. Fomin\thanks{%
  Department of Informatics, University of Bergen, N-5020 Bergen, Norway.
  \texttt{\{fedor.fomin|saket.saurabh\}@ii.uib.no}}
 \and Serge Gaspers\thanks{%
  Centro de Modelamiento Matem{\'a}tico, Universidad de Chile, 8370459 Santiago de Chile.
  \texttt{sgaspers@dim.uchile.cl} 
 }
 \and Christophe Paul\addtocounter{footnote}{-3}\footnotemark
 \and Anthony Perez\addtocounter{footnote}{-1}\footnotemark
 \and Saket Saurabh\footnotemark
 \and St{\'e}phan Thomass{\'e}\addtocounter{footnote}{-2}\footnotemark}

\date{}

\begin{document}

\maketitle

%
%
%

\begin{abstract}
A tournament $T = (V,A)$ is a directed graph in which there is exactly one arc between every pair of distinct vertices.  Given a digraph on $n$ vertices and an integer parameter $k$, the {\sc Feedback Arc Set} problem asks whether the given digraph has a set of $k$ arcs whose removal results in an acyclic digraph. The  
{\sc Feedback Arc Set} problem restricted to tournaments is known as  the {\sc $k$-Feedback Arc Set in Tournaments ($k$-FAST)} problem.  In this paper we obtain a linear vertex kernel for \FAST{}. That is, 
we give a polynomial time algorithm which given an input instance $T$ to \FAST{} obtains an equivalent instance $T'$ on $O(k)$ vertices. In fact, given any fixed $\epsilon > 0$, the kernelized instance has at most 
$(2 + \epsilon)k$ vertices. Our result improves the previous known bound of $O(k^2)$ on the kernel 
size for \FAST{}. Our kernelization algorithm solves the problem on a subclass of tournaments in polynomial time and uses a known polynomial time approximation scheme for \FAST.
\end{abstract}

\section{Introduction}
Given a directed graph $G=(V,A)$ on $n$ vertices and an integer parameter $k$, the {\sc Feedback 
Arc Set} problem asks whether the given digraph has a set of $k$ arcs whose removal results in an acyclic directed graph.
In this paper, we consider this problem in a special class of directed graphs, {\em tournaments}. A tournament $T = (V , A)$ is a directed graph in which there is exactly one directed arc between every pair of vertices. More formally the problem we consider is defined as follows. 


\medskip

\begin{quote}
\textsc{$k$-Feedback Arc Set in Tournaments} (\FAST):  Given a 
tournament $T = (V,A)$ and a positive integer $k$, does there exist a subset 
$F\subseteq A$ of at most $k$ arcs whose removal makes $T$ acyclic. 
\end{quote}

\medskip

In the weighted version of \FAST{}, we are also given integer weights (each weight is at least one) on the arcs and the objective is to find a \fas of weight at most $k$. This problem is called {\sc $k$-Weighted Feedback Arc Set in Tournaments ($k$-WFAST)}.

Feedback arc sets in tournaments are well studied from the combinatorial 
\cite{EM65,Jung70,RP70,SB61,Spencer71,Younger63}, statistical \cite{Slater61} and algorithmic \cite{ACN05,Alon06,CFR06,MathieuS07,Zyulen05,ZHJW07} points of view. The problems \FAST and \wFAST have several applications. In \emph{rank aggregation} we are given several rankings of a set of objects, and we wish to produce a single ranking that on average is as consistent as possible with the given ones, according to some chosen measure of consistency. This problem has been studied in the context of voting 
\cite{Borda1781,Condorcet1785}, machine learning \cite{CSS97}, and search engine ranking \cite{DKNS01.1,DKNS01.2}. A natural consistency measure for rank aggregation is the number of pairs that occur in a different order in the two rankings. This leads to \emph{Kemeny rank aggregation} \cite{Kemeny59,KS62}, a special case of \wFAST.

The \FAST{} problem is known to be \classNP{}-complete by recent results of Alon~\cite{Alon06} and Charbit \etal~\cite{CTY07} 
while \wFAST{} is known to be \classNP{}-complete by Bartholdi III \etal~\cite{BartholdiTT89}. From an approximation perspective, \wFAST{} is APX-hard \cite{Speckenmeyer89} but admits a 
polynomial time approximation scheme when the edge weights are bounded by a constant \cite{MathieuS07}. The problem is also well studied in parameterized complexity. In this area, a problem with input size 
$n$ and a parameter $k$ is said to be fixed parameter tractable (\classFPT{}) if there exists an algorithm to solve 
this problem in time $f(k)\cdot n^{O(1)}$, where $f$ is an arbitrary function of $k$. Raman and Saurabh~\cite{RS06}  showed that \FAST and \wFAST are \classFPT{} by obtaining an algorithm running in time $O(2.415^k \cdot k^{4.752} + n^{O(1)})$. Recently, Alon \etal~\cite{ALS09} have improved this result by giving an 
algorithm for \wFAST{} running in time $O(2^{O({\sqrt k} \log^2 k)}+ n^{O(1)})$. This algorithm runs in sub-exponential time, a trait uncommon to parameterized algorithms. In this paper we investigate \FAST\  
from the view point of kernelization, currently one of the most active subfields of parameterized algorithms. 

A parameterized problem is said to admit a {\it polynomial kernel} if there is a polynomial (in $n$) time algorithm, called a {\em kernelization} algorithm, that reduces the input instance to an instance whose size is bounded by a polynomial $p(k)$ in $k$, while preserving the answer. This reduced instance is called a {\em $p(k)$ kernel} for the problem. When $p(k)$ is a linear function of $k$ then the corresponding kernel is a linear kernel. Kernelization has been at the forefront of research in parameterized complexity 
in the last couple of years, leading to various new polynomial kernels as well as tools to show that several problems do not have a polynomial kernel under some complexity-theoretic assumptions~\cite{BodlaenderDFH08,abs-0904-0727,BousquetDTY09,DLS09,ThomasseT09}. In this paper we continue the current theme of research on kernelization and obtain a {\em linear vertex} kernel for \FAST{}. That is, 
we give a polynomial time algorithm which given an input instance $T$ to \FAST{} obtains an equivalent instance $T'$ on $O(k)$ vertices.  More precisely, given any fixed $\epsilon > 0$, we find a kernel with a most $(2 + \epsilon)k$ vertices in polynomial time. The reason we call it a linear \emph{vertex} kernel is that, even though the number of vertices in the reduced instance is at most $O(k)$, the number 
of arcs is still $O(k^2)$. Our result improves the previous known bound of $O(k^2)$ on the vertex kernel size for \FAST{}~\cite{ALS09,DGHNT06}.  For our kernelization algorithm we find a subclass of tournaments where one can find a minimum sized 
\fas\ in polynomial time (see Lemma~\ref{lem:closure}) and use the known polynomial time approximation scheme for \FAST by Kenyon-Mathieu and Schudy~\cite{MathieuS07}. The polynomial time algorithm for 
a subclass of tournaments could be of independent interest.


The paper is organized as follows. In Section~\ref{sec:prelim}, we give some definition and preliminary results regarding feedback arc sets. In Section~\ref{sec:fast} we give a linear 
vertex kernel for \FAST{}.  Finally we conclude with some remarks in Section~\ref{sec:concl}. 

\section{Preliminaries}
\label{sec:prelim}
Let $T = (V,A)$ be a tournament on $n$ vertices. We use $T_\sigma = (V_\sigma, A)$ to denote a 
tournament whose vertices are ordered under a fixed ordering $\sigma = v_1, \ldots, v_n$ (we also use $D_\sigma$ for an ordered directed graph). We say that an arc $v_iv_j$ of $T_\sigma$ is a \emph{backward arc} if $i > j$, otherwise we call it a 
\emph{forward arc}. Moreover, given any partition $\mathcal{P} := \{V_1, \ldots, V_l\}$ of $V_\sigma$, where every $V_i$ is an interval according to the ordering of $T_\sigma$, we use $A_B$ to denote all arcs between the intervals (having their endpoints in different intervals), and $A_I$ for all arcs within the intervals. 
If $T_\sigma$ contains no backward arc, then we say that it is \emph{transitive}. 

For a vertex $v \in V$ we denote its \emph{in-neighborhood} by $N^-(v):=\{u\in V \mid u v \in A\}$ and its \emph{out-neighborhood} by $N^+(v):=\{u\in V \mid v u \in A\}$.
A set of vertices $M \subseteq V$ is a \emph{module} if and only if $N^+(u) \setminus M = N^+(v) \setminus M$ for every $u,v \in M$.
For a subset of arcs $A' \subseteq A$, we define $T[A']$ to be the digraph $(V',A')$ where $V'$ is the union of endpoints of the arcs in $A'$. Given an ordered digraph $D_{\sigma}$ and an arc 
$e = v_i v_j$, $S(e) = \{v_i,\hdots, v_j\}$ denotes the \emph{span} of $e$.
The number of vertices in $S(e)$ is called the \emph{length} of $e$ and is denoted by $l(e)$.
Thus, for every arc $e = v_i v_j$, $l(e) = |i - j| + 1$. Finally, for every vertex $v$ in the span of $e$, we say that $e$ is \emph{above} $v$. 


In this paper, we will use the well-known fact that every acyclic tournament admits a transitive ordering. In particular, we will consider \emph{maximal transitive modules}. We also need the following result for our kernelization algorithm. 


\begin{lemma}(\cite{RS06})
Let $D = (V,A)$ be a directed graph and $F$ be a minimal \fas of $D$. Let $D'$ be the graph obtained from 
$D$ by reversing the arcs of $F$ in $D$, then $D'$ is acyclic.
\end{lemma}

In this paper whenever we say {\em circuit}, we mean a directed cycle. Next we introduce a definition which is useful for a lemma we prove later. 

\begin{definition}
\label{def:certificate}
Let $D_\sigma = (V_\sigma,A)$ be an ordered directed graph and let $f = vu$ be a backward arc of $D_\sigma$. We call \emph{certificate} of $f$, and denote it by $c(f)$, any directed path from $u$ to $v$ using only forward arcs in the span of $f$ in $D_\sigma$.
\end{definition}

Observe that such a directed path together with the backward arc $f$ forms a directed cycle in $D_\sigma$ whose only backward arc is $f$.


\begin{definition}
\label{def:certify}
Let $D_\sigma = (V_\sigma, A)$ be an ordered directed graph, and let $F \subseteq A$ be a set of backward arcs of $D_\sigma$. We say that we can \emph{certify} $F$ whenever it is possible to find a set $\mathcal{F} = \{c(f) : f \in F\}$ of arc-disjoint certificates for the arcs in $F$.
\end{definition}

Let $D_\sigma=(V_\sigma,A)$ be an ordered directed graph,  and let $F \subseteq A$ be a subset of backward arcs of $D_\sigma$.
We say that we can certify the set $F$ using {\em only arcs from $A'\subseteq A$} if $F$ 
can be certified by a collection ${\cal F}$ such that the union of the arcs of the certificates in $\cal F$ is contained 
in $A'$. In the following, $fas(D)$ denotes the {\em size} of a minimum feedback arc set, that is, the cardinality of a minimum sized set $F$ of arcs whose removal makes $D$ acyclic.

\begin{lemma}
\label{lem:fas}
Let $D_\sigma$ be an ordered directed graph, and let $\mathcal{P}=\{V_1,\dots, V_l\}$ be a partition of $D_\sigma$ into intervals. Assume that the set $F$ of all backward arcs of $D_\sigma[A_B]$ can be certified 
using only arcs from $A_B$. 
Then $fas(D_\sigma)=fas(D_\sigma[A_I])+fas(D_\sigma[A_B])$. Moreover, there exists a minimum sized  \fas\  of $D_\sigma$ 
containing $F$. 
\end{lemma}
\begin{proof}
For any bipartition of the arc set $A$ into $A_1$ and $A_2$, $fas(D_\sigma) \geq fas(D_\sigma[A_1])+fas(D_\sigma[A_2])$. Hence, in particular for a partition of the arc set $A$ into $A_I$ and $A_B$ we have that  $fas(D_\sigma)\geq fas(D_\sigma[A_I])+fas(D_\sigma[A_B])$. Next, we show that 
$fas(D_\sigma)\leq fas(D_\sigma[A_I])+fas(D_\sigma[A_B])$. This follows from the fact that once we reverse all the arcs in 
$F$, each remaining circuit lies in $D_\sigma[V_i]$ for some $i\in \{1,\ldots, l\}$. In other words once we reverse all the arcs in $F$, every circuit is completely contained in $D_\sigma[A_I]$. This concludes the proof of the first part of the lemma. 
In fact, what we have shown is that there exists a minimum sized  \fas\  of $D_\sigma$ containing $F$.  This concludes the proof of the lemma. 
\end{proof}





\section{Kernels for \FAST}
\label{sec:fast}
In this section we first give a subquadratic vertex kernel of size $O(k\sqrt k) $ for  \FAST\ and then improve on it 
to get our final vertex kernel of size $O(k)$. We start by giving a few reduction rules that will be needed to bound the size of the kernels.

\medskip

\begin{polyrule}
\label{rule:uselessvertex}
If a vertex $v$ is not contained in any triangle, delete $v$ from $T$.
\end{polyrule}

\medskip

\begin{polyrule}
\label{rule:distincttriangles}
If there exists an arc $uv$ that belongs to more than $k$ distinct triangles, then reverse $uv$ and decrease $k$ by $1$.
\end{polyrule}

\medskip

We say that a reduction rule is {\em sound}, if whenever the rule is applied to an instance $(T,k)$ 
to obtain an instance $(T', k')$, $T$ has a \fas\ of size at most $k$ if and only if  $T'$ has a \fas\ of size at most $k'$. 

\begin{lemma}(\cite{ALS09,DGHNT06})
Rules~\ref{rule:uselessvertex} and~\ref{rule:distincttriangles} are sound and can be applied in polynomial time.
\end{lemma}

The Rules~\ref{rule:uselessvertex} and~\ref{rule:distincttriangles} together led to a quadratic kernel for \wFAST~\cite{ALS09}. Earlier, these rules were 
used by Dom \etal~\cite{DGHNT06} to obtain a quadratic kernel for \FAST. We now add a new reduction rule that will allow us to obtain the claimed bound on the kernel sizes for \FAST. 
Given an ordered tournament $T_\sigma=(V_\sigma,A)$, we say that $\mathcal{P} = \{V_1,\ldots,V_l\}$ is a \emph{safe partition} of $V_\sigma$ into intervals whenever it is possible to certify the backward arcs of $T_\sigma[A_B]$ using only 
arcs from $A_B$.

\medskip

\begin{polyrule}
\label{rule:safepartition}
Let $T_\sigma$ be an ordered tournament, $\mathcal{P} = \{V_1,\ldots,V_l\}$ be a safe partition of $V_\sigma$ into intervals and $F$ be the set of backward arcs of $T_\sigma[A_B]$. 
Then reverse all the arcs of $F$ and decrease $k$ by $|F|$. 
\end{polyrule}

\medskip

\begin{lemma}
Rule~\ref{rule:safepartition} is sound.
\end{lemma}
\begin{proof}
Let $\mathcal{P}$ be a safe partition of $T_\sigma$. Observe that it is possible to certify all the backward arcs, that is $F$, using only arcs in $A_B$. Hence using Lemma~\ref{lem:fas} we have that 
$fas(T_\sigma) = fas(T_\sigma[A_I]) + fas(T_\sigma[A_B])$.  Furthermore, by Lemma~\ref{lem:fas} we also know that 
there exists a minimum sized  \fas\  of $D_\sigma$ containing $F$.  Thus, $T_\sigma$ has a \fas of size at most $k$ if and only if the tournament $T_\sigma '$ obtained from $T_\sigma$ by reversing all the arcs of $F$ 
has a \fas of size at most $k-|F|$.
\end{proof}

\subsection{A subquadratic kernel for \FAST}

In this section, we show how to obtain an $O(k\sqrt k)$ sized vertex kernel for \FAST. To do so, we introduce the following reduction rule. 

\medskip

\begin{polyrule}
\label{rule:reverse}
Let $V_m$ be a \emph{maximal transitive module} of size $p$, and $I$ and $O$ be the set of in-neighbors and out-neighbors of the vertices of $V_m$ in $T$, respectively. 
Let $Z$ be the set of arcs $uv$ such that $u\in O$ and $v\in I$. If $q = |Z| < p$ then reverse all the arcs in $Z$ 
and decrease $k$ by $q$.
\end{polyrule}

\medskip


\begin{figure}[h!]
\begin{center}
	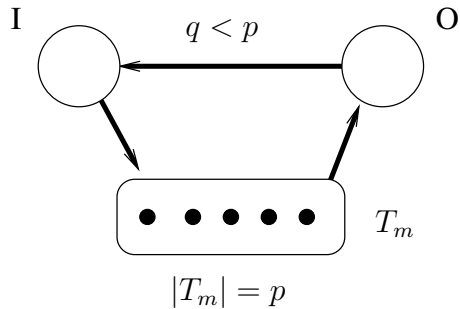
\end{center}
\caption{A transitive module on which Rule~\ref{rule:reverse} applies.}
\end{figure}

\begin{lemma}
Rule~\ref{rule:reverse} is sound and can be applied in linear time.
\end{lemma}

\begin{proof}
We first prove that the partition ${\cal P}=\{I,V_m,O\}$ forms a safe partition of the input tournament. Let $V_m'=\{w_1,\ldots,w_q\}\subseteq V_m$ be an arbitrary subset of size $q$ of $V_m$ 
and let $Z=\{u_i v_i~|~1\leq i\leq q\}$. Consider the collection ${\cal F}=\{v_i w_i u_i~|~u_i v_i\in Z, w_i\in V_m'\}$ and notice that it certifies all the arcs 
in $Z$. In fact we have managed to certify all the backwards arcs of the partition using only arcs from $A_B$ and hence ${\cal P}$ forms a safe partition. Thus, by Rule~\ref{rule:safepartition}, 
it is safe to reverse all the arcs from $O$ to $I$. The time complexity follows from the fact that computing a modular decomposition tree can be done in $O(n + m)$ time 
on directed graphs~\cite{MM05}.
\end{proof}

We show that any \textsc{Yes}-instance to which none of the Rules \ref{rule:uselessvertex},~\ref{rule:distincttriangles} and \ref{rule:reverse} could be applied has at most $O(k\sqrt k)$ vertices.

\begin{theorem}
\label{thm:FAST}
Let $(T = (V,A),k)$ be a \textsc{Yes}-instance to \FAST{} which has been reduced according to Rules \ref{rule:uselessvertex},~\ref{rule:distincttriangles} and \ref{rule:reverse}. Then  
$T$ has at most $O(k\sqrt k)$ vertices.
\end{theorem}

\begin{proof}
Let $S$ be a \fas\ of size at most $k$ of $T$ 
and let $T'$ be the tournament obtained from $T$ by reversing all the arcs in $S$. Let $\sigma$ be the transitive ordering of $T'$ and $T_\sigma=(V_\sigma,A)$ be the ordered tournament corresponding to the ordering $\sigma$. We say that a vertex is \emph{affected} if it is incident to some arc in $S$. Thus, the number of affected vertices is at most $2|S|\leq 2k$. 
The reduction Rule~\ref{rule:uselessvertex} ensures that the first and last vertex of $T_\sigma$ are 
affected. To see this note that if the first vertex in $V_\sigma$ is not affected then it is a source vertex (vertex with in-degree $0$) and hence it is not part of any triangle and thus Rule~\ref{rule:uselessvertex} would have applied. We can similarly argue for the last vertex. 
Next we argue that there is no backward arc $e$ of length greater than $2k+2$ in $T_\sigma$. Assume to the contrary that $e=u v$ is a backward arc with $S(e)=\{v,x_1,x_2,\ldots,x_{2k+1},\ldots,u\}$ and hence 
$l(e)>2k+2$. Consider the collection 
${\cal T}=\{v x_i u~|~1\leq i \leq 2k\}$ and observe that at most $k$ of these triples can contain an arc from 
$S\setminus\{e\}$ and hence there exist at least $k+1$ triplets in $\cal T$ which corresponds to distinct triangles 
all containing $e$. But then $e$ would have been reversed by an application of Rule~\ref{rule:distincttriangles}. Hence,  
we have shown that there is no backward arc $e$ of length greater than $2k+2$ in $T_\sigma$. 
 Thus $\sum_{e \in S} l(e) \le 2k^2+2k$. 
 
We also know that between two consecutive affected vertices there is exactly one maximal transitive module. Let us denote by $t_i$ the number of vertices in these modules, where $i\in \{1,\ldots, 2k - 1\}$. The objective here is to bound the number of vertices in $V_\sigma$ or $V$ using $\sum_{i=1}^{2k - 1} t_i$. To do so, observe that since $T$ is reduced under the Rule~\ref{rule:reverse}, there are at least $t_i$ backward arcs above every module with $t_i$ vertices, each of length at least $t_i$. This implies that $\sum_{i = 1}^{2k - 1} t_i^2 \le \sum_{e \in S} l(e) \le 2k^2+2k$. Now, using the Cauchy-Schwarz inequality we can show the following. 
\begin{eqnarray*}
	\sum_{i=1}^{2k - 1} t_i  =  \sum_{i=1}^{2k - 1} t_i \cdot 1 
			     \le  \sqrt{\sum_{i=1}^{2k - 1} t_i^2 \cdot \sum_{i=1}^{2k - 1} 1} 
			    \le  \sqrt{(2k^2+2k) \cdot (2k - 1)} 
			     =  \sqrt{4k^3+2k^2-k}.
\end{eqnarray*}
Thus every reduced \textsc{Yes}-instance has at most $ \sqrt{4k^3+2k^2-k}+ 2k=O(k\sqrt k)$ vertices.
\end{proof}

\subsection{A linear kernel for \FAST}
\label{section:linear}

We begin this subsection by showing some general properties about tournaments which will be useful in obtaining a linear kernel 
for \FAST.

\subsubsection{Backward Weighted Tournaments}
Let $T_\sigma$ be an ordered tournament with weights on its backward arcs. We call such a tournament a \emph{backward weighted tournament} and denote it by $T_\omega$, and use $\omega(e)$ to denote the weight of a backward arc $e$. For every interval $I:= [v_i, \ldots, v_j]$ we use $\omega(I)$ to denote the total weight of all backward arcs having both their endpoints in $I$, that is, $\omega(I)=\sum_{e=uv}w(e)$ where $u,v\in I$ and $e$ is a backward arc.   


\begin{definition}{\bf (Contraction)} Let $T_\omega=(V_\sigma,A)$ be an ordered tournament with weights on its backward arcs and $I = [v_i, \ldots, v_j]$ be an interval. The contracted tournament is defined 
as $T_{\omega'}=(V_{\sigma'}=V_\sigma\setminus \{I\} \cup \{c_I \},A')$. The arc set  $A'$ is defined as follows.
\begin{itemize}
\item It contains all the arcs $A_1=\{uv~|~uv\in A, u\notin I,v \notin I \}$
\item Add $A_2=\{uc_I~|~uv\in A, u\notin I, v\in I \}$ and $A_3=\{c_Iv~|~uv\in A, u\in I, v\notin I \}$. 
\item Finally, we remove every forward arc involved in a $2$-cycle after the addition of arcs in the previous step. 
\end{itemize}
The order $\sigma'$ for $T_{\omega'}$ is provided by $\sigma'=v_1,\ldots,v_{i-1},c_I,v_{j+1},\ldots ,v_n$.  
We define the weight of a backward arc $e=xy$  of $A'$ as follows.  
\begin{equation*} 
w'(xy) = \left\{ 
\begin{array}{rl} 
w(xy) & \text{if } xy \in A_1\\ 
\sum_{\{xz\in A~|~z\in I\}}w(xz) & \text{if } xy\in A_2\\ 
\sum_{\{zy\in A~|~z\in I\}}w(zy) & \text{if } xy\in A_3 
\end{array} \right. 
\end{equation*} 
We refer to Figure $2$ for an illustration.

\end{definition}

\begin{figure}[h!]
\begin{center}
	\includegraphics[width=10cm]{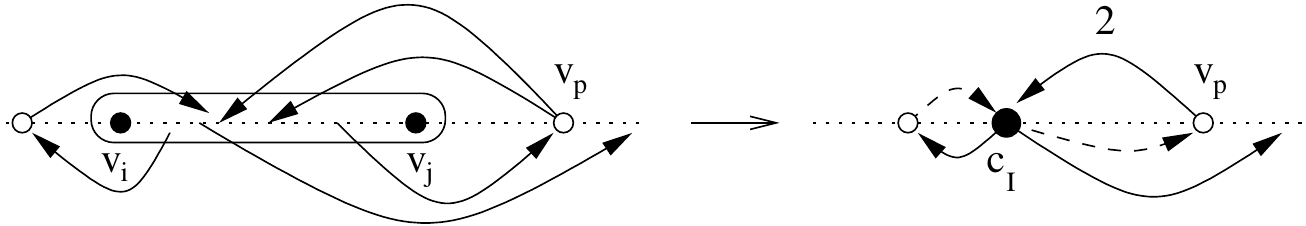}
\end{center}
\caption{Illustration of the contraction step for the interval $I := [v_i, \ldots, v_j]$.\label{fig:contraction}}
\end{figure}

Next we generalize the notions of certificate and certification  (Definitions~\ref{def:certificate} and~\ref{def:certify}) to backward weighted tournaments. 

\begin{definition}
\label{def:certificateweight}
Let $T_\omega = (V_\sigma, A)$ be a backward weighted tournament, and let $f = vu \in A$ be a 
backward arc of $T_\omega$. We call \emph{$\omega$-certificate} of $f$, and denote it by ${\cal C}(f)$, a collection of  
$\omega(f)$ arc-disjoint directed paths going from $u$ to $v$ and using only forward arcs in the span of $f$ in $T_\omega$.
\end{definition}



\begin{definition}
\label{def:certifyweighted}
Let $T_\omega = (V_\sigma, A)$ be a backward weighted tournament, and let $F \subseteq A$ be a 
subset of backward arcs of $T_\omega$. We say that we can \emph{$\omega$-certify} $F$ whenever it is possible to
find a set $\mathcal{F} = \{{\cal C}(f) : f \in F\}$ of arc-disjoint $\omega$-certificates for the arcs in $F$.
\end{definition}

\begin{lemma}
\label{lem:closure}
Let $T_\omega=(V_\sigma,A)$ be a backward weighted tournament such that for every interval $I := [v_i, \ldots, v_j]$ the following holds:

\begin{align}
 \label{eq:one}
	2 \cdot \omega(I) &\le |I| - 1\enspace
\end{align}
Then it is possible to $\omega$-certify the backward arcs of $T_\omega$.
\end{lemma}

\begin{proof}
 Let $V_\sigma=v_1,\ldots,v_n$. 
The proof is by induction on $n$, the number of vertices. 
Note that by applying \eqref{eq:one} to the interval $I = [v_1, \hdots, v_n]$, we have that there exists a vertex $v_i$ in $T_\omega$ that is not incident to any backward arc. Let $T{_\omega'} = (V{_\sigma'}, A')$ denote the tournament $T_\omega \setminus \{v_i\}$. We say that an interval $I$ is \emph{critical} whenever $|I|\geq 2$ and $2 \cdot \omega(I)=|I| - 1$. We now consider several cases, based on different types of critical intervals.
\begin{enumerate}[(i)]
	\item \label{enum:one} Suppose that there are no critical intervals. Thus, in $T'_\omega$, every interval satisfies~(\ref{eq:one}), and hence by induction on $n$ the result holds.
	\item \label{enum:two} Suppose now that the only critical interval is  $I = [v_1, \hdots, v_n]$, and let $e = vu$ be a backward arc above $v_i$ with the maximum length. Note that since $v_i$ does not belong to any backward arc, we can use it to form a directed path $c(e)=uv_iv$, which is a certificate for $e$. We now consider $T{_\omega'}$ where the weight of $e$ has been decreased by $1$. In this process if $\omega(e)$ becomes $0$ then we reverse the arc $e$. We now show that every interval of $T{_\omega'}$ respects~\eqref{eq:one}. If an interval $I' \in T{_\omega'}$ does not contain $v_i$ in the corresponding interval in $T_\omega$, then  by our assumption we have that $2 \cdot \omega(I') \le |I'| - 1$. Now we assume that the interval  corresponding to $I'$ in $T_\omega$ contains $v_i$ but either $u\notin I'\cup\{v_i\}$ or  $v\notin I'\cup\{v_i\}$. Then we have $2 \cdot \omega(I') = 2 \cdot \omega(I) < |I| - 1 = |I'| $ and hence we get that  
$2 \cdot \omega(I') \leq |I'| - 1$. Finally, we assume that the interval  corresponding to $I'$ in $T_\omega$ contains $v_i$ 
and  $u,v\in I'\cup\{v_i\}$. 
In this case, $2 \cdot \omega(I') = 2 \cdot (\omega(I) - 1) \leq |I| - 1 - 2 < |I'| - 1$. Thus, by the induction hypothesis, we obtain a family of arc-disjoint $\omega$-certificates $\mathcal{F'}$ which $\omega$-certify the backward arcs of $T_\omega'$. Observe that the maximality of $l(e)$ ensures that if $e$ is reversed then it will not be used in any $\omega$-certificate of $\mathcal{F'}$, thus implying that $\mathcal{F'} \cup c(e)$ is a family $\omega$-certifying the backward arcs of $T_\omega$.

\item \label{enum:three} Finally, suppose that there exists a critical interval $I \subsetneq V_\sigma$. Roughly speaking, we will show that $I$ and $V_\sigma \setminus I$ can be certified separately. To do so, we first show the following.

\emph{Claim}. Let $I \subset V_\sigma$ be a critical interval. Then the tournament $T_{\omega'}=(V_{\sigma'},A')$ obtained from $T_\omega$ by contracting $I$ satisfies the conditions of the lemma.

\begin{proof}
Let $H'$ be any interval of $T_{\omega'}$. As before if $H'$ does not contain $c_I$ then the result holds by hypothesis. Otherwise, let $H$ be the interval corresponding to $H'$ in $T_\omega$. We will show that $2\omega(H') \le |H'| - 1$. By hypothesis, we know that $2\omega(H) \le |H| - 1$ and that $2\omega(I) = |I| - 1$. Thus we have the following. 
\begin{eqnarray*}
	2 \omega(H')  =  2 \cdot (\omega(H) - \omega(I)) 
		       \le  |H| - 1 - |I| + 1 
		       =  (|H| + 1 - |I|) - 1 
		       =  |H'| - 1
\end{eqnarray*}
Thus, we have shown that the tournament $T_{\omega'}$ satisfies the conditions of the lemma.
\end{proof}
We now consider a minimal critical interval $I$. By induction, and using the claim, we know that we can obtain a family of arc-disjoint $\omega$-certificates $\mathcal{F'}$ which $\omega$-certifies the backward arcs of $T_{\omega'}$ without using any arc within $I$. Now, by minimality of $I$, we can use~(\ref{enum:two}) to obtain a family of arc-disjoint $\omega$-certificates $\mathcal{F''}$ which $\omega$-certifies the backward arcs of $I$ using only arcs within $I$. Thus, $\mathcal{F'} \cup \mathcal{F''}$ is a family $\omega$-certifying all backward arcs of $T_\omega$.
\end{enumerate}
This concludes the proof of the lemma. 
\end{proof}

In the following, any interval that does not respect condition \eqref{eq:one} is said to be a \emph{dense interval}. 

\begin{lemma}
\label{lem:safepartition}
Let $T_\omega=(V_\sigma,A)$ be a backward weighted tournament with $ |V_\sigma| \ge 2p + 1$ and $\omega(V_{\sigma}) \le p$.
Then there exists a safe partition of $V_\sigma$ with at least one backward arc between the intervals and it can be computed in polynomial time.
\end{lemma}

\begin{proof}
The proof is by induction on $n=|V_\sigma|$. Observe that the statement is true for $n = 3$,  which is our base case.

For the inductive step, we assume first that there is no dense interval in $T_\omega$. In this case Lemma~\ref{lem:closure} ensures that 
the partition of $V_\sigma$ into singletons of vertices is a safe partition. So from now on we assume that there exists 
at least one dense interval. 

Let $I$ be a dense interval. By definition of $I$, we have that $\omega(I) \geq \frac{1}{2} \cdot |I|$. We now contract $I$ and obtain the  backward weighted tournament $T_{\omega'} = (V_{\sigma'}, A')$. In the contracted tournament $T_{\omega'}$,  we have:
\begin{align*}
	\begin{cases}
			|V_{\sigma'}| &\geq 2p + 1 - (|I| - 1) = 2p - |I| + 2;\\
			\omega'(V_{\sigma'}) &\leq p - \frac{1}{2} \cdot |I|.
	\end{cases}
\end{align*}
\noindent
Thus, if we set $r := p - \frac{1}{2} \cdot |I|$, we get that $|V_{\sigma'}| \geq 2r + 1$ and 
$\omega'(V_{\sigma'}) \leq r$. Since $|V_{\sigma'}|<|V_{\sigma}|$, by the induction hypothesis we can find 
a safe partition $\mathcal{P}$ of $T_{\omega'}$, and thus obtain a family $\mathcal{F}_{\omega'}$  
that $\omega$-certifies the backward arcs of $T_{\omega'}[A_B]$ using only arcs in $A_B$. 

We claim that $\mathcal{P}'$ obtained from $\cal P$ by substituting $c_I$ by its corresponding interval $I$ is 
a safe partition in $T_\omega$. To see this, first observe that if $c_I$ has not been used to $\omega$-certify the backward arcs in $T_{\omega'}[A_B]$, that is,  
$c_I$ is not an end point of any arc in the $\omega$-certificates, then we are done. So from now on we assume that $c_I$ has been part of a $\omega$-certificate for some backward arc. Let $e$ be a backward arc in 
$T_{\omega'}[A_B]$, and let $c_{\omega'}(e) \in \mathcal{F}_{\omega'}$ be a $\omega$-certificate of $e$. 
First we assume that $c_I$ is not the first vertex of the certificate $c_{\omega'}(e)$ (with respect to ordering $\sigma'$),  and let $c_1$ and $c_2$ be the left (in-) and right (out-) neighbors of $c_I$ in $c_{\omega'}(e)$. By definition of the contraction 
step together with the fact that there is a forward arc between $c_1$ and $c_I$ and between $c_I$ and $c_2$ in 
$T_{\omega'}$, we have that there were no backward arcs between any vertex in the interval corresponding to $c_I$ and $c_1$ and $c_2$ in the original tournament $T_\omega$. So we can always find a vertex in $I$ to replace $c_I$ in 
$c_{\omega'}(e)$, thus obtaining a certificate $c(e)$ for $e$ in $T_{\omega}[A_B]$ (observe that $e$ remains a 
backward arc even in $T_\omega$). Now we assume that $c_I$ is either a first or last vertex in the certificate 
$c_{\omega'}(e)$. Let $e'$ be an arc corresponding to $e$ in $T_{\omega'}$  with one of its endpoints being 
$e_I \in I$. To certify $e'$ in $T_\omega[A_B]$, we need to show that we can construct a certificate 
$c(e')$ using only arcs of $T_\omega[A_B]$. We have two cases to deal with. 
\begin{enumerate}[(i)]
	\item If $c_I$ is the first vertex of $c_{\omega'}(e)$ then let $c_1$ be its right neighbor in $c_{\omega'}(e)$. Using the same argument 
	as before, there are only forward arcs between any vertex in $I$ and $c_1$. In particular, there is a forward 
	arc $e_Ic_1$ in $T_\omega$, meaning that we can construct a $\omega$-certificate for $e'$ in $T_\omega$ by setting 
	$c(e') := (c_{\omega'}(e) \setminus \{c_I\}) \cup \{e_I\}$.

\begin{figure}[h!]
\begin{center}
	\includegraphics[width=10cm]{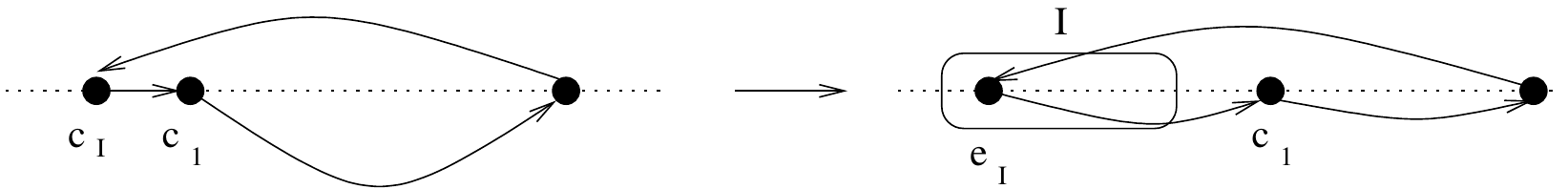}
\caption{On the left, the $\omega$-certificate $c_{\omega'}(e) \in \mathcal{F}_{\omega'}$. On the right, the corresponding $\omega$-certificate obtained in $T_\omega$ by replacing $c_I$ by the interval $I$.}
\label{fig:wcontraction}
\end{center}
\end{figure}

	\item If $c_I$ is the last vertex of $c_{\omega'}(e)$ then let $c_q$ be its left neighbor in $c_{\omega'}(e)$. Once again, we have that there are only forward arcs between $c_q$ and vertices in $I$, and thus between $c_q$ and $e_I$. So using this 
	we can construct a $\omega$-certificate for $e'$ in $T_\omega$.
\end{enumerate}
Notice that the fact that all $\omega$-certificates are pairwise arc-disjoint in $T_{\omega'}[A_B]$ implies that the corresponding $\omega$-certificates are arc-disjoint in $T_\omega[A_B]$, and so $\mathcal{P}'$ is indeed a safe partition of $V_\sigma$.
\end{proof}

We are now ready to give the linear size kernel for \FAST. To do so, we make use of the fact that there exists a polynomial time approximation scheme 
for this problem~\cite{MathieuS07}.
\begin{theorem}
For every fixed $\epsilon > 0$, there exists a vertex kernel for \FAST\ with at most $(2 + \epsilon)k$ vertices that can be computed in polynomial time. 
\end{theorem}

\begin{proof}
Let $(T = (V,A),k)$ be an instance of \FAST. For a fixed $\epsilon > 0$, we start by computing a \fas\ $S$ of size 
at most $(1 + \frac{\epsilon}{2})k$. To find such a set $S$, we use the known polynomial time approximation scheme 
for \FAST~\cite{MathieuS07}. Then, we order $T$ with the transitive ordering of the tournament obtained by reversing every arc of $S$ in $T$. Let $T_\sigma$ denote the resulting ordered tournament. By the upper bound on the size of $S$, we know that $T_\sigma$ has at most $(1 + \frac{\epsilon}{2})k$ backward arcs. Thus, if $T_\sigma$ has more than $(2 + \epsilon)k$ vertices then Lemma~\ref{lem:safepartition} ensures that we can find a safe partition with at least one backward arc between the intervals in polynomial time. Hence we can reduce the tournament  by 
applying Rule~\ref{rule:safepartition}. We then apply Rule~\ref{rule:uselessvertex}, and repeat the previous steps 
until we do not find a safe partition or $k = 0$. In the former case, we know by Lemma~\ref{lem:safepartition} that $T$ can have at most $(2 + \epsilon)k$ vertices, thus implying the result. In all other cases we return \textsc{No}.
\end{proof}

\section{Conclusion}
\label{sec:concl}
In this paper we obtained linear vertex kernel for \FAST{}, in fact, a vertex kernel of size $(2+\epsilon)k$ for any fixed $\epsilon >0$. 
The new bound on the kernel size improves the previous known bound of $O(k^2)$ on the vertex kernel size for \FAST{} given in ~\cite{ALS09,DGHNT06}. It would be interesting to see if one can obtain kernels for other problems using either polynomial time approximation schemes or a constant factor approximation algorithm for the corresponding problem.  An interesting problem which remains unanswered is, whether there exists a linear or even a $o(k^2)$ vertex kernel for the {\sc $k$-Feedback Vertex Set in Tournaments ($k$-FVST)} problem. In the 
{\sc $k$-FVST} problem we are given a tournament $T$ and a positive integer $k$ and the aim is to find a set of at most $k$ vertices whose deletion makes the input tournament acyclic. The smallest known kernel for {\sc $k$-FVST} 
has size $O(k^2)$.

\bibliographystyle{plain}
\bibliography{fast}

\end{document}